\documentclass[12pt]{article}
 \usepackage{calc,graphics}
\usepackage{hyperref} \usepackage{microtype}
\usepackage
[disable]
{todonotes}
\usepackage{
  booktabs,natbib}
\usepackage{times} \usepackage{bm}
\usepackage[plain,noend]{algorithm2e}
\usepackage{amsgen,amsmath,amssymb,latexsym} \usepackage{graphicx}
\usepackage{enumerate} \usepackage{type1cm} \usepackage{xspace}
\usepackage{tikz} \usetikzlibrary{shapes,arrows,shadows}
\usetikzlibrary{decorations.pathmorphing}

\newcommand{\pr}{'} \newcommand{\newpar}{{\vspace{0.3cm} \noindent}}
\newcommand{\indep}{\mbox{$\,\perp\!\!\!\perp\,$}}

\newcommand{\ie}{{\em i.e.\/}\xspace}  
 \newcommand{\etc}{{\em
    etc.\/}\xspace} \newcommand{\halm}{\hspace*{\fill} $\Box$\par}

\newtheorem{theorem}{Theorem}[section]
\newtheorem{definition}{Definition}[section]
\newtheorem{condition}{Condition} \newtheorem{cor}{Corollary}
\newcommand{\corref}[1]{\mbox{Corollary~\ref{cor:#1}}}
\newcommand{\Equationref}[1]{\mbox{Equation~(\ref{#1})}}
\newcommand{\Figureref}[1]{\mbox{Figure~\ref{#1}}}
\newcommand{\figref}[1]{\mbox{Figure~\ref{#1}}}

\newcommand{\Definitionref}[1]{\mbox{Definition~\ref{def:#1}}}
\newcommand{\defref}[1]{\mbox{Definition~\ref{def:#1}}}

\newcommand{\condref}[1]{\mbox{Condition~\ref{#1}}}

\newcommand{\theoremref}[1]{\mbox{Theorem~\ref{#1}}}

\newcommand{\secref}[1]{\mbox{\S$\,$\ref{#1}}}
\newcommand{\E}{\mbox{E}} 
\newcommand{\rosso}{}
\newcommand{\nero}{\color{black}\xspace}

\newtheorem{expl}{Example}
\newenvironment{ex}{\begin{expl}\rm}{\halm\end{expl}}

\newenvironment{proof}{\noindent {\bf Proof.
  }}{\halm\vspace{\baselineskip}}

\DeclareGraphicsExtensions{.pdf,.png,.gif,.jpg,.tif,.ps}

\usepackage{filecontents}

\begin{document}

\title{\bf Stochastic Mechanistic Interaction}
 \author{Carlo Berzuini\thanks{Centre for Biostatistics, The University of Manchester, University
Place, Manchester M139PL, U.K. (carlo.berzuini@manchester.ac.uk)} \and
    A. Philip Dawid\thanks{Statistical Laboratory, Centre for
      Mathematical Sciences, University of Cambridge, Wilberforce Road, Cambridge CB30WB, U.K.
(apd@statslab.cam.ac.uk)}}
\maketitle

  \begin{abstract}
    \noindent We propose a fully probabilistic formulation of the
    notion of mechanistic interaction between the effects of putative
    causal factors $A$ and $B$ in producing an outcome event $Y$. We
    define mechanistic interaction in terms of departure from a
    generalized ``noisy OR'' model, under which the multiplicative
    causal effect of $A$ (resp., $B$) on the probability of positive
    outcome ($Y = 1$) cannot be enhanced by manipulating $B$ (resp.,
    $A$).We present conditions under which mechanistic interaction in
    the above sense can be assessed via simple tests on excess risk or
    superadditivity, in a possibly retrospective regime of
    observation.  These conditions are defined in terms of conditional
    independence relationships that can often be checked on a
    graphical representation of the problem. Inference about
    mechanistic interaction between direct, or path-specific, causal
    effects can be accommodated in the proposed framework.  The method
    is illustrated with the aid of a study in
    experimental psychology.

  \end{abstract}

  \section{Introduction}
  \label{Introduction}

  Consider an outcome event whose probability responds to
  manipulations of two variables, $A$ and $B$.  We are interested in
  whether the effects of $A$ and $B$ interact in producing the event
  in some fundamental mechanistic sense. For example, we might be
  interested in whether an environmental exposure $A$ interferes with
  the effect of a drug $B$ on a disease at some
  mechanistic---presumably molecular---level.  Such a relationship,
  which we shall make more formal in a later section of this paper, we
  call {\em mechanistic interaction\/}.

  \newpar Let the binary variable $Y$ indicate positive ($Y=1$) or
  negative ($Y=0$) outcome.  One might begin to investigate
  mechanistic interaction by fitting a regression model of the
  dependence of $Y$ on $(A,B)$ and then testing for presence of
  statistical $A \times B$ interaction, but such a test will depend on
  the chosen response scale, and will generally not be interpretable
  in any deep mechanistic sense.  Hence the need for a mathematical
  formalization of mechanistic (as opposed to statistical)
  interaction, and of the conditions under which this phenomenon can
  be detected from empirical data via appropriate, response-scale
  independent, statistical tests.  In many applications, discovery of
  mechanistic interaction could represent a step forward in the
  understanding of the studied system.  In genetics, evidence of
  mechanistic interaction between two genes with respect to a
  phenotype of interest could point to the molecular mechanisms
  implicated~\citep{Bernardinelli2012}.

  \newpar Ideally we would wish to assess mechanistic interaction by a
  controlled experiment, but this is often not possible or not
  convenient. Various authors have proposed tests for inferring
  mechanistic interaction (suitably defined) from observational data
  \citep{Rothman1976,Rothman1998,Greenland1988,
    Skrondal2003,VanderWeeleRobins2008, VanderWeeleRobinsAnnals,
    VanderWeele2009, VanderWeeleApplications2010,
    VanderWeeleBiometrika2010,VanderWeeleNature2011,
    VanderWeeleContinuous, VanderWeeleLaird2011}.  Consider, for
  example, the case where $A$ and $B$ are binary, and let $C$ denote a
  further (possibly empty) set of observed variables.  Let $R_{abc}$
  denote the observational risk of a positive outcome, $Y=1$,
  conditional on $A=a$, $B=b$, $C=c$.  Then, in certain observational
  situations, and under certain conditions, the following properties
  (of which the first is stronger than the second) have been shown to
  imply some form of mechanistic interaction:

  \vspace{-0.6cm}

\begin{align}
  \intertext{\em Excess risk:}
  \label{excessrisk}
  R_{11c}-R_{10c}-R_{01c} &> 0.  \\
  \intertext{\em Superadditivity:}
  \label{superadditivity}
  R_{11c}-R_{10c}-R_{01c} + R_{00c} &> 0.
\end{align}

\noindent These can be alternatively expressed as:
\begin{align}
  \intertext{\em Excess risk:}
  \label{Sexcess risk}
  S_{10c}+S_{01c}-S_{11c} &> 1,\\
  \intertext{\em Superadditivity:}
  \label{Ssuperadditivity}
  S_{10c}+ S_{01c} - S_{11c}- S_{00c} &> 0, \end{align} \noindent
where $S_{ijc} := 1 - R_{ijc}$ is the corresponding probability of
negative outcome (of $Y=0$). More precisely, the above conditions give
criteria for {\em synergistic\/} mechanistic interaction between $A$
and $B$ in {\em producing\/} the outcome event, in that the combined
effect (suitably measured) of increases in $A$ and in $B$ to increase
the probability of a positive outcome is greater than expected on the
basis of their individual effects. This is the interpretation we shall
maintain here. The case of {\em antagonistic\/} mechanistic
interaction, where the the combined effect is smaller than expected,
is readily handled by interchanging the values $0$ and $1$ for $Y$,
and interchanging the $R$s and $S$s in
equations~\eqref{excessrisk}--\eqref{Ssuperadditivity} and
elsewhere. An important property of the above tests is that they are
(at least approximately, under assumptions) testable under
retrospective sampling.

\newpar Most work to date on mechanistic interaction has been
formulated assuming the potential outcome (PO) framework
\citep{Rubin1974} for causality or some essentially equivalent
formulation, though the literature also offers some purely
probabilistic approaches. The former category is exemplified by the
stochastic PO approach of \citet{VanderWeeleStochastic}; the latter is
exemplified by previous work of the authors of this paper
\citep{BerzuiniDawid2012} and by the recent paper of
\citet{Ramsahai2013}. Section~\ref{Related work} discusses these
approaches and their
limitations. 

\section{Assumptions and notation}
\label{notation}

We are interested in the way the probability \rosso distribution \nero
of $Y$ would react to real or hypothetical manipulations of causal
factors $A$ and $B$, and in particular whether or not the effects of
$A$ and $B$ on $Y$ can be regarded as interacting in some fundamental
mechanistic sense.  In order to address this, we must first understand
what might be meant by ``no mechanistic interaction''. Here we suggest
a possible explication of this concept.  This however is not absolute,
but relative to an appropriately chosen ``context''.  That is, we
specify certain {\em context variables\/} $W$ in the problem, which
might modify in some way the dependence of $Y$ on $(A,B)$, and only
consider this dependence within a fixed context, \ie, conditional on
fixed values $W = w$.  Different choices of the variables in $W$ for
different causal questions are possible in the same problem.

\newpar In contrast to the formulation of \citet{BerzuiniDawid2012},
we do not require that $Y$ be a deterministic function of $(A, B,W)$,
but allow for a fully stochastic dependence of $Y$ on these
inputs. Note that this allows considerable freedom in the selection of
the context variables. Indeed, even in those rare cases when there
does exist a choice for $W$ supporting a deterministic relationship,
this might be regarded as inhabiting too deep a level of description
to be useful for the purpose at hand, and a more coarse-grained
choice, yielding a genuinely stochastic relationship, could be more
appropriate. In any given application, care must be taken to ensure
that we are arguing at a suitable level of granularity. As an analogy,
for most purposes it is appropriate to think of the determination of
the sex of an embryo as governed by a random process (essentially a
fair coin toss), even though a very detailed description of the
positions, motions, properties and behaviours of the gametes prior to
fertilisation might allow deterministic prediction.

\newpar Our definition of ``no mechanistic interaction'' will relate
to a (possibly hypothetical) ``interventional regime'', in which the
values of $A$ and $B$ are set by some external agent or
process. However, the data available to investigate this property will
generally have been generated under some other, typically purely
observational, regime, where, in particular, the values of $A$ and $B$
have arisen in some uncontrolled stochastic way. We will need to be
able to relate these regimes in order to transfer information from one
to the other. To streamline this task we introduce the {\em regime
  indicator\/} $\sigma_{AB}$, a non-stochastic variable, where
$\sigma_{AB} = ab$ indicates the interventional regime where $A$ is
set to $a$ and $B$ to $b$, and $\sigma_{AB} = \emptyset$ the
observational regime. More generally, $\sigma_{X}$ will denote a
similar regime indicator for interventions on a set $X$ of variables
of interest.

\newpar We introduce the symbol
\begin{equation}
  \label{pi}
  \pi_{w}(a,b) := P(Y=0 \mid  W=w, \sigma_{AB}= ab),
\end{equation}
\noindent for the probability of a negative outcome in context $W=w$,
when $A$ and $B$ are set to respective values $a$ and $b$ by an
exogenous intervention.  Then one way---which we shall henceforth
adopt---of understanding the effect $(A,B)$ on $Y$ is by considering
the dependence of $\pi_{w}(a,b)$ on $(a,b)$. We shall measure the
effect exerted on $Y$ by a change of the values set for $(A,B)$ from
$(a^{'},b^{'})$ to $(a,b)$ by the {\em relative probability of
  negative outcome} (RPNO),
\begin{equation*}
  \frac{\pi_{w}(a,b) }{ \pi_{w}(a^{'},b^{'})},
\end{equation*}
\rosso
\noindent with the value $1$ representing ``no effect''.  \nero

\subsection{Structural conditions}
\label{Assumptions}

We shall impose the following structural conditions:

\begin{condition}
  \label{struct1}
  The causal factors $A$ and $B$ are continuous or ordered categorical
  scalar random variables.
\end{condition}
\begin{condition}
  \label{struct2}
  The effects of $A$ and $B$ on $Y$ are positive: for any fixed
  $(b,w)$ (resp., $(a,w)$), $P (Y = 1 \mid W=w, \sigma_{AB}=ab)$ is a
  non-decreasing function of $a$ (resp., $b$).
\end{condition}
An alternative expression of \condref{struct2} is that $\pi_{w}(a,b)$,
given by \eqref{pi}, is, for each $w$, a non-increasing function of
each of $a$ and $b$.

\section{No mechanistic interaction}
\label{Independent effects}

We henceforth make the structural assumptions of the previous section.

\subsection{Point null hypothesis}
\label{secpoint}
One possible way of expressing the concept of {\em no mechanistic
  interaction between $A$ and $B$ in producing\/} the outcome event is
that, for all $w$, we can express $\pi_{w}(a,b)$ in the product form
\begin{equation}
  \label{pointnull}
  \pi_{w}(a,b) = \lambda_{w}(a) \mu_{w}(b)
\end{equation} for all $a\in{\cal A}$, $b\in{\cal B}$.
We term this the {\em point null hypothesis\/}.

\newpar \Equationref{pointnull} can alternatively be expressed as
requiring:

\newpar For all $a,a\pr\in{\cal A}$, $b,b\pr\in{\cal B}$,
\begin{equation}
  \label{pointnull1}
  {\pi_{w}(a,b)}{\pi_{w}(a\pr,b\pr)} =
  {\pi_{w}(a,b\pr)}{\pi_{w}(a\pr,b)}
\end{equation}
or, when the denominators are positive,
\begin{equation}
  \label{pointnull2}
  \frac{\pi_{w}(a,b)}{\pi_{w}(a\pr,b)}  =
  \frac{\pi_{w}(a,b\pr)}{\pi_{w}(a\pr,b\pr)}
\end{equation}
or
\begin{align}
  \pi_{w}(a,b) \; &= \frac{\pi_{w}(a\pr,b) \; \pi_{w}(a,b\pr)}
  {\pi_{w}(a\pr,b\pr)}.
  \label{pointnull3}
\end{align}

\noindent When both $A$ and $B$ are binary this is equivalent to the
single requirement
\begin{equation}
  \label{leaky}
  \pi_{w}(1,1) = \pi_{w}(1,0)\pi_{w}(0,1)/\pi_{w}(0,0).
\end{equation}
In the special case that $\pi_{w}(0,0)=1$, this becomes
\begin{equation}
  \label{noisy}
  \pi_{w}(1,1) = \pi_{w}(1,0)\pi_{w}(0,1).
\end{equation}
Imposing the further requirements $\pi_{w}(1,0) = \pi_{w}(0,1) = 0$
would now imply
\begin{equation}
  \label{OR}
  \pi_{w}(1,1) = 0,
\end{equation}
and this constellation of values represents $Y$ as the Boolean
expression \mbox{$A$ OR $B$}.  The intermediate case of
$\pi_{w}(0,0)=1$ together with \eqref{noisy} is the ``noisy OR''
generalisation of this, while the general expression \eqref{leaky} is
``leaky noisy OR'' \citep{Pearl1988,Lemmer2004, Zagorecki2004}.

\newpar \rosso The point null hypothesis models the negative outcome
as the result of two uncertain causes failing to produce $Y=1$, the
first cause failing with a probability that depends on $W$ and on the
value we force on $A$, and the second cause failing with a probability
that depends on $W$ and on the value we force on $B$.  \nero
imagine Player 1 being assigned a ball of size $A=a$ and Player 2 a
ball of size $B=b$, and then each player being invited to knock his
respective pin down.  Let us call it a positive event ($Y=1$) when at
least one of the two players knocks the pin over, and let $Y=0$
indicate instead that no pin is knocked over.  Think of $\pi_w(a,b)$
as representing the probability of $Y=0$ with assigned ball sizes
$(a,b)$ and with $W=w$ indicating specific circumstances such as air
humidity, temperature.  Structural condition \ref{struct2} will be
satisfied here if we assume that each player's ability to knock the
pin down will not decrease on being assigned a bigger ball. In the
context of this example the point null hypothesis, as expressed by
Equation \eqref{pointnull}, asserts that the probability of $Y=0$ is
the product of the probability $\lambda_w(a)$ that Player 1 fails to
knock the pin down, and the probability $\mu_w(b)$ that Player 2
fails. The mechanistic interpretation of this being that the
performance of one player is not affected by the size of the ball
given to the other player. In the context of this example, Equation
\eqref{pointnull2} explicitly states that the RPNO effect of giving
one player a larger ball will not be changed by giving a larger ball
also to the other one.


\subsection{Interval null hypothesis}
\label{secint}
Taking into account that we are only interested in synergistic (as
opposed to antagonistic) interaction, we can weaken the above point
null hypothesis, as expressed by \eqref{pointnull2} or
\eqref{pointnull3}, to yield the following {\em interval null
  hypothesis\/}:

\newpar For all $w$, and all $a\ge a\pr\in{\cal A}$, $b\geq
b\pr\in{\cal B}$, we have
\begin{equation}
  \label{intnull2}
  \frac{\pi_{w}(a,b)}{\pi_{w}(a\pr,b)}  \geq
  \frac{\pi_{w}(a,b\pr)}{\pi_{w}(a\pr,b\pr)}.
\end{equation}
\rosso
\noindent On account of \condref{struct2}, both sides of the
inequality are $ \le 1$, as they represent the RPNO effect of
increasing $A$ while keeping $B$ fixed at a particular value. A large
departure from $1$ means the effect is strong.  The inequality states
that an interventional increase in $B$ will not strengthen the effect
of an interventional increase in $A$. The interpretation holds with
 the roles of $A$ and $B$ interchanged.  \nero
\noindent We may equivalently express the interval null hypothesis as
\begin{align}
  \pi_{w}(a,b) \; &\geq \frac{\pi_{w}(a\pr,b) \; \pi_{w}(a,b\pr)}
  {\pi_{w}(a\pr,b\pr)}.
  \label{intnull3}
\end{align}

\subsection{Mechanistic interaction}
\label{mechint}
\begin{definition}
  \label{def:mechint}
  We say that the causal factors $A$ and $B$ exhibit {\em mechanistic
      interaction\/}, or that their effects {\em interfere\/}, in {\em
      producing\/}   a positive outcome,   when the interval null
  hypothesis (and so {\em a   fortiori\/} the point null hypothesis)
  fails: that   is to say, when there exists a value $w$ of the
  context variable $W$ and values $a > a'$  for $A$, and $b > b'$ for
  $B$, such that
  \begin{equation}
      \label{eqinter}
      \pi_w(a,b)\pi_w(a',b') < \pi_w(a',b)\pi_w(a,b').
  \end{equation} \end{definition} \noindent When this holds, we write $A * B \,[W]$, \rosso or $A*B$ when $W \equiv \emptyset$. \nero This notation makes it explicit that the property is relative to the specified context variable $W$. \nero 

\newpar Because we focus on synergy, we have defined ``interference to
produce''; we could similarly define its antagonism counterpart,
``interference to prevent''.

\newpar If Equation~(\ref{eqinter}) holds then clearly $\pi_{w}(a',b)
0$, $\pi_{w}(a,b') >0$.  Also $\pi_w(a',b') >0$ for, if it were not
so, all terms of the equation would be 0 by virtue of
~\condref{struct2}.  Thus \Definitionref{mechint} applies just when
there exist $w$, $a>a'$, $b>b'$ such that
\begin{equation}
  \label{synergism 5}
  \frac{\pi_w(a,b)}{\pi_{w}(a',b)} <
  \frac{\pi_w(a,b')}{\pi_{w}(a',b')}.
\end{equation}




\newpar Inequality~(\ref{eqinter}) represents a stochastic extension
of the deterministic mechanistic interaction concept of
\cite{BerzuiniDawid2012}.  Under such deterministic dependence of $Y$
on $(A,B,W)$, each term in (\ref{eqinter}) can only take values $0$ or
$1$.  Together with Condition~\ref{struct2}, this implies:
\begin{equation}
  \label{extension}
  \pi_w(a,b) = 0 < \pi_{w}(a',b) =
  \pi_w(a,b') = \pi_{w}(a',b') = 1.
\end{equation}
\noindent The above inequality says that there are values $b, b' \in
{\cal B}$ such that, in some context $W=w$, when we set $B = b$ a
manipulation of $A$ from $a'$ to $a$ causes $Y$ to change from $0$ to
$1$; whereas, in the same context, when we set $B=b'$, the same
manipulation makes no difference to $Y$.  In other words, whenever $Y$
is deterministic, presence of mechanistic interaction in our
formulation is characterized by the fact that we can prevent a certain
manipulation of $A$ from producing the outcome event through an
appropriate action on $B$; and {\em vice versa\/}.  If further $A$ and
$B$ are binary, then $(a,a',b,b')= (1,0,1,0)$, and the definition says
that $A$ and $B$ interact mechanistically in producing $Y=1$ when
there exists a value $w$ of the context variable $W$ such that the
dependence of $Y$ on $(A,B)$ obeys the Boolean conjunction law: $Y =A
\wedge B$.

\section{Observational identification of mechanistic interaction}
\label{obsid}

We now consider how we might use observational data to assess the
presence or absence of mechanistic interaction between the effects of
$A$ and $B$ on $Y$.  We shall do this be means of a set $C \subseteq
W$ of observed context variables; the remaining variables $U =
W\setminus C$ may be observed or unobserved.

\newpar We shall need to consider, in addition to the structural
conditions of \secref{Assumptions}, some {\em causal conditions\/},
relating the behaviours under observational and interventional
circumstances.  These we express as follows, where we have used the
symbol $\indep$ for ``conditionally independent
of''~\citep{Dawid1979,Dawid2002}.

\begin{condition}
  \label{caus1}
  $Y \indep \sigma_{AB} \mid (A,B,W)$.
\end{condition}

\begin{condition}
  \label{caus2}
  $U \indep (A,B,\sigma_{AB}) \mid C$.
\end{condition}

\newpar Finally, we shall sometimes require observational independence
between $A$ and $B$, conditional on $C$:
\begin{condition}
  \label{caus3} $A \indep B \mid (C, \sigma_{AB})$.
\end{condition}

\newpar \condref{caus1} requires that the effects of $A$ and $B$ on
$Y$ be ``unconfounded'', conditional on the context variables
$W$. \condref{caus2} says that, conditional on $C$, the distribution
of $U$ is fixed: the same under both interventional and observational
conditions and, further, independent of the values of $A$ and
$B$. \condref{caus3} holds trivially for an interventional regime
$\sigma_{AB} = ab$, so only has force for the observational regime
$\sigma_{AB} = \emptyset$. It is a strong condition, but in certain
circumstances can be avoided---see \corref{nonindep} below.

\subsection{Causal diagrams}
\label{Causal diagrams}

\todo[inline]{I have re-written this section}

It will sometimes be possible to represent a coherent set of causal
and conditional independence assumptions about the problem by using an
extension \citep{Dawid2002} of the methodology of directed acyclic
graphs \citep{Cowell1999}.  The extension, called an {\em influence
  diagram\/} (ID), incorporstes the relevant regime indicators in the
graph as decision nodes sending arrows into the variables they relate
to. The resulting ID expresses conditional independence relationships
between problem variables and regime indicators.  These independencies
can be read off the graph with the aid of a graphical criterion such
as {\em $d$-separation\/} \citep{geiger90} or its {\em moralisation\/}
equivalent \citep{lauritzen90}.  By so doing, we can check whether the
required conditions for the validity of our interaction tests are
satisfied.

\newpar \Figureref{fig1}, for example, might represent the effects of
genetic variants $A$ and $B$ on myocardial infarction $Y$, possibly
mediated by obesity $M$, with $G$ representing a set of observed
descriptors (diet, social status, \etc) of socio-economic status.  In
the same diagram, the regime indicators $\sigma_{AB}$ and $\sigma_{M}$
specify the regimes under which the values of $(A,B)$ and of $M$,
respectively, are generated, be it by passive observation or by
intervention.


\newpar For a problem that can be modelled by the ID of
\Figureref{fig1}, causal conditions \ref{caus1}--\ref{caus3} follow by
application of the moralisation criterion to the graph if we choose
the sets $W, U$ and $C$ to be empty.

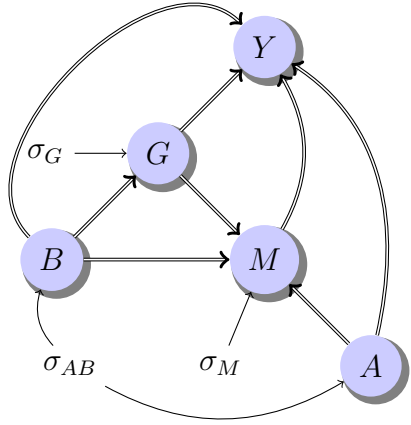
\begin{figure}[bt]
  \centering
  \begin{tikzpicture}[align=center,node distance=2cm]
      \node (Y) [circular drop shadow,decorate, fill=blue!20,circle]
      {$Y$};   \node (Z) [circular drop shadow,below left
    of=Y,decorate,   fill=blue!20,circle] {$G$};   \node (M) [circular
    drop shadow,below right of=Z,   decorate, fill=blue!20,circle]
    {$M$};   \node (A)[circular drop shadow,below right of=M,decorate,
      fill=blue!20,circle] {$A$};   \node (B) [circular drop
    shadow,below left of=Z,   decorate, fill=blue!20,circle] {$B$};
      \node (SM) [left of=A] {$\sigma_{M}$};   \node (SG) [left of=Z,
    node distance=1.5cm] {$\sigma_{G}$};   \node (SAB) [left of=SM]
    {$\sigma_{AB}$};
      \draw[->,double, bend right] (M) to node {} (Y);
      \draw[->,double](Z) to node {} (Y);
      \draw[->,double] (Z) to node {} (M);   \draw[->,double] (A) to
    node {} (M);   \draw[->,double, bend right,   in=-140] (A) to node
    {} (Y);   \draw[->,double, bend left, out=90,   in=90] (B) to node
    {} (Y);   \draw[->,double] (B) to node {} (M); \draw[->,double]
    (B) to node {} (Z);   \draw[->] (SM) to node {} (M);  \draw[->]
    (SG) to node {} (Z);   \draw[->, bend right] (SAB) to node {} (A);
      \draw[->, bend left] (SAB) to node {} (B);
  \end{tikzpicture}
  \caption{\small In an influence diagram such as this, variables may
    depend on their predecessors in the graph in a fully stochastic
    way. Regime nodes, here $\sigma_{AB}$, $\sigma_{M}$ and
    $\sigma_{G}$, determine whether the variables into which they send
    arrows are manipulated (interventional regime) or observed
    (observational regime). This diagram is also used in Examples
    \ref{ex:Example1}-3}
  \label{fig1}
\end{figure}

\subsection{Main theorem}
\label{sec:main}

Our observational criterion for mechanistic interaction will involve a
dichotomisation of the ranges of $A$ and $B$, determined by respective
``cutoff thresholds'' $\tau_A$ and $\tau_B$.  Let $\alpha$ be the
indicator variable of ``$A > \tau_A$'', and $\beta$ the indicator
variable of ``$B > \tau_B$''.  The symbol $R_{ijc}$ is henceforth
reinterpreted as:
\begin{align}
  R_{ijc} &= P(Y=1 \mid \alpha=i, \beta=j,             C=c,
  \sigma_{AB}= \emptyset),
  \label{reinterpretation}
\end{align}
\noindent and likewise $S_{ijc} = 1-R_{ijc}$.  We reinterpret the
inequalities (\ref{excessrisk})--(\ref{Ssuperadditivity})
correspondingly. Note that $R_{ijc}$ is estimable from data on
variables $A$, $B$ and $C$, gathered under the observational regime.

\begin{theorem}
  \label{main}
  Assume Conditions~\ref{struct1}  and \ref{struct2}, and that
  Conditions~ \ref{caus1}--\ref{caus3} hold in stratum $C=c$.  \rosso
  Assume also the validity of the {\em uniform positivity condition},
  that $\pi_{cu}(\tau_A,\tau_B) > 0$ for all $u$ (or
  $\pi_{c}(\tau_A,\tau_B) > 0$, when $U$ is empty).  \nero Then the
  presence of superadditivity in stratum $C=c$ implies $A*B\, [W]$.
\end{theorem}

\begin{proof}
  We proceed by assuming both superaddivity and the interval null
  hypothesis, and deriving a contradiction.  In the following, all
  probabilities and expectations are taken under the observational
  regime $\sigma_{AB} = \emptyset$.

  \newpar Using \condref{caus2}, we have, for $i, j = 0, 1$:
  \begin{displaymath}
      S_{ijc} = \E \left(S_{ijcU} \mid C = c\right)
  \end{displaymath}
  where
  \begin{displaymath}
      S_{ijcu} := P(Y = 0 \mid \alpha = i, \beta = j, C = c, U= u).
  \end{displaymath}
  Thus
  \begin{equation}
      \label{uu}
      S_{10c}+ S_{01c} - S_{11c}- S_{00c} = \E\left(S_{10cU}+ S_{01cU} - S_{11cU}- S_{00cU} \mid C = c\right).
  \end{equation}
  \noindent Let $U^* := \{u: S_{10cu}+ S_{01cu} - S_{11cu}- S_{00cu} >
  0\}$.  On account of \eqref{uu} and \eqref{Ssuperadditivity}, $P(U^*
  \mid C = c) > 0$; in particular, $U^*$ is non-empty.  Fix any $u^*
  \in U^*$, and define $S^*_{ij} := S_{ijcu^*}$, $\pi^*(ab) :=
  \pi_{cu^*}(a,b)$.  Then
  \begin{equation}
      \label{ssss}
      S^*_{11}+ S^*_{00} -   S^*_{10} - S^*_{01} < 0.
  \end{equation}

  \noindent Now by virtue of Conditions~\ref{caus1} and \ref{caus2},
  \begin{equation}
      \label{iniziale}
      S^*_{ij}  = \E\{\pi^*(A,B) \mid \alpha = i, \beta=j, C=c\}.
  \end{equation}

  \noindent Let $a > \tau_A$ and $b > \tau_B$.  By virtue of the
  interval null hypothesis \eqref{intnull3} we then have
  \begin{equation}
      \label{comp}
      {\pi_0}\,\pi^*(a,b) \geq \pi^*( a,\tau_B)  \pi^*( \tau_A, b)
  \end{equation}
  where $\pi_0 := \pi^*(\tau_A,\tau_B)>$ 0 by assumption; hence
  \begin{eqnarray}
      \label{prod1}
      \pi_0 S^*_{11} &\geq&  \E\{\pi^*( A,\tau_B)  \pi^*( \tau_A, B)
      \mid \alpha = 1, \beta=1, C=c\}\\
      \nonumber
      &=& \E\{\pi^*( A,\tau_B)
      \mid \alpha = 1,  C=c\}\\
      \nonumber
      &&{}\times\E\{ \pi^*( \tau_A, B)
      \mid \beta=1, C=c\}
  \end{eqnarray}
  on using \condref{caus3}.  That is, defining
  \begin{eqnarray*}
      S_A &: =& \E\left\{\pi^*(A,\tau_B) \mid  \alpha = 1, C=
          c\right\}\\    S_{\overline A} &: =&
      \E\left\{\pi^*(A,\tau_B) \mid  \alpha = 0, C=
          c\right\}\\    S_B &: =& \E\left\{\pi^*(\tau_A,B) \mid  \beta = 1, C=
          c\right\}\\
      S_{\overline B} &: =& \E\left\{\pi^*(\tau_A,B) \mid \beta = 0, C= c\right\},
  \end{eqnarray*}
  we have
  \begin{align}
      \label{ineq11}
      \pi_0 S^*_{11} &\geq S_A S_B.\\
      \intertext{Similarly,}   \label{ineq00}
      \pi_0 S^*_{00} &\geq S_{\overline A} S_{\overline B},\\
      \intertext{while}   \label{ineq01}
      \pi_0 S^*_{01} &\leq S_{\overline A} S_B,\\
      \label{ineq10}   \pi_0 S^*_{10} &\leq S_A S_{\overline B}.
  \end{align}

  \noindent From these inequalities and \eqref{ssss}, we obtain $(S_A
  - S_{\overline A})(S_B - S_{\overline B}) < 0$.  But by virtue of
  \condref{struct2}, we have $S_A \leq \pi_0 \leq S_{\overline A}$,
  $S_B \leq \pi_0 \leq S_{\overline B}$, yielding $(S_A - S_{\overline
     A})(S_B - S_{\overline B}) \geq 0$. This contradiction proves the
  theorem.
\end{proof}

\begin{cor}
  \label{cor:nonindep}
  Suppose that both $A$ and $B$ are binary.  Then the conclusion of
  \theoremref{main} holds even if we remove the independence
  requirement of \condref{caus3}.
\end{cor}
\begin{proof}
  In this case $\tau_A = \tau_B = 0$, and \eqref{prod1} becomes
  \begin{eqnarray*}
      \pi_0 S^*_{11} &\geq&  \pi^*( 1,0)  \pi^*( 0, 1)\\
      &=& S_A S_B
  \end{eqnarray*}
  so that \eqref{ineq11}, and similarly \eqref{ineq00},
  \eqref{ineq01}, \eqref{ineq10}, continue to hold even without
  assuming independence.  The rest of the proof is unchanged.
\end{proof}

\newpar For the next Corollary we introduce the following weaker form
of \condref{struct2}, appropriate for cases where the directionality
of the effect is not known {\em a priori\/}:
\begin{condition}
  \label{struct2a}
  The effect of $A$ on $Y$ is either {\em positive\/}, in the sense
  that $P(Y =1 \mid W=w, \sigma_{AB}= ab)$ \color{black} is a
  non-decreasing function of $a$ for all $(b,w)$; or {\em negative\/},
  in the sense that $P(Y =1 \mid W=w, \sigma_{AB}= ab)$ is a
  non-increasing function of $a$ for all $(b,w)$; and similarly with
  the r\^oles of $A$ and $B$ interchanged.
\end{condition}

\begin{cor}
    Suppose that, in the statement of \theoremref{main} or
    \corref{nonindep}, we replace \condref{struct2} by the weaker
    \condref{struct2a}, and at the same time replace the
    superadditivity property \eqref{superadditivity} by the stronger
    excess risk property \eqref{excessrisk} (again reinterpreted in
    terms of definition \eqref{reinterpretation}).  Then the
    conclusion remains valid.
\end{cor}

\begin{proof}
  We use the same notation as in the proof of \theoremref{main}.
  Arguing similarly to that proof, we deduce that there exists a value
  $u^*$ of $U$ such that $R^*_{11} - R^*_{10} - R^*_{01} > 0$.  This
  implies both
  \begin{align}
      \label{x10}
      S^*_{11} - S^*_{10} &< 0\\
      \intertext{and}   \label{x01}   S^*_{11} - S^*_{01} &< 0,
  \end{align}
  as well as \eqref{uu}.

  \newpar From \eqref{x10}, \eqref{ineq11} and \eqref{ineq10}, we
  deduce $S_A(S_B - S_{\overline B}) < 0$, whence $S_B < S_{\overline
    B}$.  But if the effect of $B$ were negative we would have $S_B
  \geq S_{\overline B}$.  Hence the effect of $B$ is
  positive. Similarly, using \eqref{x01}, we deduce that the effect of
  $A$ is positive.  The rest of the proof now follows as before.
\end{proof}

\noindent {\bf Comment}.  By allowing the dichotomization of $A$ to be
arbitrary, the above theorem fits the common situation where the
continuous factor is made available in a dichotomized form, without
the possibility of recovering the original continuous measurements.

\todo[inline]{I have written and added the following theorem. I am not
  sure whether it should be kept in the paper, or whether its
  implications in some way incorporated.}

\rosso

\begin{theorem}
  \label{aggiuntoDaCarlo}
  Assume that the conditions of Theorem~\ref{main} hold, with
  Conditions~\ref{caus1}--\ref{caus3} satisfied by some choice  $(C
  \equiv C^{*}, U \equiv U^{*})$ of variables $C$ and $U$.  Then the
  theorem conditions remain satisfied if we replace the choice $(C
  \equiv C^{*},  U \equiv U^{*})$ with  $(C \equiv C^{*}, U \equiv
  \emptyset)$.
\end{theorem}

\begin{proof}
  We shall use the following axiomatic properties \citep{Dawid1979} of
  the conditional independence relationship:
  \begin{quote}
    \begin{description}
    \item[{\rm Decomposition:}] $X \indep Y \,\mid\, Z$ and $W=f(Y)
      \,\, \rightarrow \,\, X \indep W \, \mid \, Z$
    \item[{\rm Weak union:}] $X \indep Y \, \mid \, Z$ and $W =
      f(Y)\,\, \rightarrow \,\, X \indep Y \, \mid \, (W,Z)$
    \item[{\rm Contraction:}]  $X \indep Y \mid Z$ and $ X \indep W
       \mid  (Y,Z) \, \rightarrow$ $ X \indep (Y,W) \mid Z$.
    \end{description}
  \end{quote}
  \noindent where we write $W = f(Y)$ to mean that $W$ is a function
  of $Y$.  By weak union, causal \condref{caus2} implies $U^{*} \indep
  (A,B,\sigma_{AB}) \,\mid \, (A,B,C^{*})$ which, in turn, implies
  \begin{equation}
    \label{implicata1}
    U^{*} \indep \sigma_{AB} \,\mid \, (A,B,C^{*})
  \end{equation}
  \noindent by decomposition.  By contraction, property
  (\ref{implicata1}) and causal \condref{caus1} jointly imply
  $(Y,U^{*}) \indep \sigma_{AB} \,\mid \, (A,B,C^{*})$ which, by
  decomposition, implies
  \begin{equation}
    \label{implicata3}
    Y \indep \sigma_{AB} \,\mid \, (A,B,C^{*}).
  \end{equation}
  \noindent Replacing the choice $(C \equiv C^{*}, U \equiv U^{*})$
  with $(C \equiv C^{*}, U \equiv \emptyset)$ leads to \condref{caus1}
  taking the form \eqref{implicata3}--- which we have just shown to
  hold, while it makes \condref{caus2} vacuous and leaves
  \condref{caus3} unaffected.  This completes the proof.
\end{proof}

\newpar It follows from Theorem~\ref{aggiuntoDaCarlo} that, if we wish
to check causal conditions~\ref{caus1} and \ref{caus2} for some choice
of context $W = (C,U)$, where $C$ is observed, we could first check
condition~\ref{caus2} for the simpler case of context $W=C$.  If it
does not hold in this case, we know the conditions can not hold for
any choice of $U$.


\section{Direct effects interaction}
\label{Direct effects interaction}

This section of the paper examines relationships between mechanistic
interaction and mediation. Mediation analysis hinges on the concept of
{\em direct effect\/} of a variable $X$ on $Y$. One variant of this
concept, the {\em direct effect of $X$ on $Y$ controlling for $F$\/},
is meant to quantify the sensitivity of $Y$ to changes in $X$ when $F$
is held fixed by intervention, that is, when a (perhaps hypothetical)
physical intervention changes the value of $X$ from some reference
value $x$ to some value $x^{'}$, while $F$ is set to some constant
\citep{pearldir,Robins1992}.  \newpar We connect our theory to the
theory of mediation by defining the concept of mechanistic interaction
between $A$ and $B$ when a further variable $F$, which could itself be
affected by $A$ and/or $B$, is set by intervention to a constant.  Let
$Z := W \setminus F$ denote the unmanipulated context variables, and
extend the notation \eqref{pi} by writing
\begin{eqnarray}
  \label{piectended}
  \pi_{z}^{f}(a,b) &:= P(Y=0 \mid  Z=z, \sigma_{AB}= ab, \sigma_F=f),
\end{eqnarray}
\noindent for the probability of $Y = 0$ given $Z=z$, conditional on
$A$, $B$ and $F$ being manipulated to take on values, $a$, $b$ and
$f$, respectively.  Take the the direct effect of $A$ on $Y$
controlling for $(B=b,F=f)$, in context $Z=z$, to be measured in terms
of relative probability of negative outcome by the quantity
\begin{eqnarray}
  \label{causalrsp}
  \frac{\pi^f_{z}(a,b)} {\pi^f_{z}(a',b)}.
\end{eqnarray}
\noindent We say there is {\em no direct mechanistic interaction\/} if
the act of setting $B$ to a higher value can never enhance the direct
effect of $A$, as measured by \eqref{causalrsp} with $a > a'$.  This
leads to the following generalization of our previous \defref{mechint}
of mechanistic interaction:

\begin{definition}
  \label{def:mechintgeneral}

  We say that the causal factors $A$ and $B$ interact mechanistically
  to produce \rosso $Y=1$ \nero under $F$-intervention if there exist
  values $(f,a > a',b > b',z)$ for $(F,A,B,Z)$, respectively, such
  that

\begin{equation}
    \label{generalized interaction}
    \frac{\pi^f_{z}(a,b)}{\pi^f_{z}(a',b)}
    \,  <  \, \frac{\pi^f_{z}(a,b')}{\pi^f_{z}(a',b')}.
\end{equation}
\noindent In this case we write $A * B \,\mid\, F \, [Z]$, or $A * B\,
\mid\, F =f \, [Z]$ if interest focuses a specific value $f$ imposed
on $F$.  We alternatively describe condition \eqref{generalized
  interaction} as mechanistic interaction between the direct effects
of $A$ and $B$ on $Y$, controlling for $F$, in context $Z=z$.
\end{definition}

\newpar The following theorem holds.

\begin{theorem}
  \label{thm:direct}
  Suppose
  \begin{condition}
      \label{struct7}   $Y \; \indep \; \sigma_{F} \mid
    (A,B,F,Z,\sigma_{AB})$.
  \end{condition}
  (That is, conditional on $(A,B,F,Z)$, the dependence of $Y$ on $F$
  is not further affected by the way the value of $F$ has been
  generated, be it by mere observation or by intervention.)

  \newpar Then $A * B\, \mid \, F\, [Z]$ if and only if $A * B \,
  [F,Z]$.
\end{theorem}
\begin{proof}
  In this case $\pi^f_z(a,b) = \pi_{fz}(a,b)$.
\end{proof}

\newpar Thus under \condref{struct7} we can use theorem \ref{main} and
its Corollaries to investigate mechanistic interaction between direct
effects.

\section{Examples}
\label{Examples}

\todo[inline]{Examples section entirely rewritten}

\newpar We now illustrate our framework with the aid of the following
examples.

\begin{ex}
  \label{ex:Example1} Brader and colleagues \citep{brader2008} study
  the reaction of public opinion to media stories about immigration.
  White adult participants were invited to read a mock newspaper story
  illustrating the costs of immigration.  The story was the same
      for all individuals, except for the stated ethnicity (latino
  {\em vs.\/}\ white) of the immigrants.  Brader (see also
  \citep{3589}) found that whites are more likely to oppose
  immigration when the story involves latino (rather than white)
  immigrants.

  \newpar We take Brader's study to be described by the ID of
  \figref{fig1}, with $B$ representing the participant's age, $A$
  indicating whether the participant was randomized to a ``latino'' or
  to a ``white'' story, $M$ 
  representing the participant's level of ``anxiety about
  immigration'', as measured through a questionnaire administered at
  the end of the reading, $G$ representing observed socio-educational
  variables and $Y=1$ indicating a positive answer to the question:
  ``Do you agree about sending the Congress a letter of complaint
  about immigration?''. The ID acknowledges the influence of
  socio-educational variables on both anxiety and outcome.

  \newpar The question whether the total effects of age and ethnical
  story framing interfere with each other can be addressed on the
  basis of Brader's data under the assumptions of \figref{fig1}, and
  assuming that the structural conditions and the uniform positivity
  condition are also valid. Suffices to show that with the choice $W
  \equiv C \equiv U \equiv \emptyset$, the causal conditions for the
  identifiability of $A*B$ hold.  In fact, with that choice,
  \condref{caus2} becomes vacuous and the remaining causal conditions
  for the identifiability of $A*B$ take forms $Y \indep \sigma_{AB}
  \mid (A,B)$ and $A \indep B \mid \sigma_{AB}$, both of which hold in
  \figref{fig1}.  We conclude that the superadditivity condition
  $R_{11}-R_{01}-R_{10}+R_{00}>0$ (or its excess risk equivalent if
  appropriate) is a valid basis for a population-wide test of $A*B$ in
  Brader's study.
\end{ex}

\begin{ex}
  \label{ex:Example2} We shall now continue our analysis of Brader's
  study. We shall continue to take \figref{fig1} as a valid
  representation of the problem, and to assume the validity of the
  structural and uniform positivity conditions.  Under these
  assumptions, we shall now show that the question whether $A$ and $B$
  interact within a specific socio-epidemiological stratum \rosso (any
  problem with the fact this latter is a post-treatment variable)
  \nero can be addressed on the basis of the study data.  Suffices to
  show that with the choices $W \equiv C \equiv G$ and $U \equiv
  \emptyset$, the causal conditions for the identifiability of $A*B\,
  [G]$ hold.  In fact, with that choice, \condref{caus2} becomes
  vacuous and the remaining causal conditions for the identifiability
  take forms $Y \indep \sigma_{AB} \mid (A,B,G)$ and $A \indep B \mid
  (G,\sigma_{AB})$, both of which hold in \figref{fig1}.  We conclude
  that the presence of superadditivity in a $G=g$ stratum of the
  population,
  \begin{equation}
    R_{11g}-R_{01g}-R_{10g}+R_{00g}>0
    \label{esempio2}
  \end{equation}
  \noindent in Brader's study implies $A*B\, [G]$.  Furthermore,
  \condref{struct7} for the equivalence of $A*B\, [G]$ and $A*B \,
  \mid \, G$ takes the form $Y \indep \sigma_G \, \mid \,
  (A,B,G,\sigma_{AB})$, which holds in \figref{fig1}.  It follows
  that, in Brader's study, evidence in favour of \eqref{esempio2} (or
  of the corresponding excess risk condition, if appropriate) will
  support the hypothesis of an interaction between the direct effects
  of $A$ and $B$ on $Y$ unmediated by $G$.

\end{ex}

\begin{ex}
  \label{ex:Example 3} Consider the class of problems described by
  \figref{fig1}, and take all variables to be observed, with $A$, $B$
  and $Y$ binary. Consider the question whether the effects of $A$ and
  $B$ interact mechanistically within a stratum $M=m, G=g$ of the
  population.  This question can be addressed on the basis of
  the data in this example.  To see this consider that, with the
  choices $W \equiv (G,M) \equiv C$ and $U \equiv \emptyset$, the
  causal conditions for the identifiability of $A*B \, [M,G]$, are
  satisfied. This is because, with those choices, \condref{caus2}
  becomes vacuous, \condref{caus3} does not apply because $A$ and $B$
  are binary, and Condition \ref{caus1} takes the form $Y \indep
  \sigma_{AB} \, \mid \, (A,B,G,M)$, which follows from \figref{fig1}.
  We conclude that, whenever the remaining (structural and uniform
  positivity) conditions for identifiability hold, the presence of
  superadditivity in a $(M=m,G=g)$ stratum of the population,
  \begin{equation}
    R_{11mg}-R_{01mg}-R_{10mg}+R_{00mg}>0,
    \label{esempio3}
  \end{equation}
  \noindent implies $A*B \, [M,G]$ in this example.

  \newpar Furthermore, \condref{struct7} for the equivalence of $A*B
  \, [M,G]$ and $A*B \, \mid \, (M,G)$ takes the form $Y \indep
  (\sigma_M,\sigma_G) \, \mid \, (A,B,G,M,\sigma_{AB})$, which follows
  from \figref{fig1}.  Hence, in this example, evidence of
  superadditivity in the sense of \eqref{esempio3} (or of excess risk,
  if appropriate) will corroborate the hypothesis of a direct effects
  interaction between the effects of $A$ and $B$ on $Y$, unmediated by
  $(M,G)$: $A*B \, \mid \, (M,G)$.  However in the more general case
  where $A$ or $B$ are non-binary, it will not be possible to ignore
  causal condition \ref{caus3}, which is in fact violated in this
  example.
\end{ex}

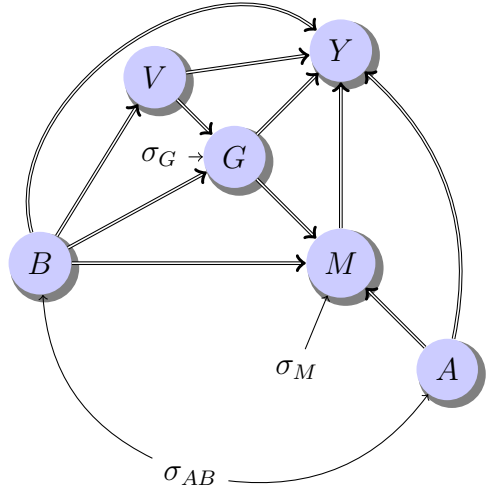
\begin{figure}[bt]
  \centering
  \begin{tikzpicture}[align=center,node distance=2cm]
      \node (Y) [circular drop shadow,decorate, fill=blue!20,circle]
      {$Y$};   \node (Z) [circular drop shadow,below left
    of=Y,decorate,   fill=blue!20,circle] {$G$};   \node (M) [circular
    drop shadow,below right of=Z,   decorate, fill=blue!20,circle]
    {$M$};   \node (A)[circular drop shadow,below right of=M,decorate,
      fill=blue!20,circle] {$A$};   \node (B) [circular drop
    shadow,left of=M,    node distance=4cm,   decorate,
    fill=blue!20,circle] {$B$};   \node (SM) [left of=A]
    {$\sigma_{M}$};    \node (SG) [left of=Z,node distance=1cm]
    {$\sigma_{G}$};   \node (SAB) [below left of=SM] {$\sigma_{AB}$};
      \node (V) [circular drop shadow,above left of=Z,decorate,
       node distance=1.5cm,   fill=blue!20,circle] {$V$};
      \draw[->,double] (M) to node {} (Y);   \draw[->,double] (B) to
    node {} (Z);   \draw[->,double] (Z) to node {} (Y);
      \draw[->,double] (V) to node {} (Z);   \draw[->,double] (V) to
    node {} (Y);   \draw[->,double] (Z) to node {} (M);
      \draw[->,double] (A) to node {} (M);   \draw[->,double, bend
    right,   in=-150] (A) to node {} (Y); \draw[->,double] (B) to node
    {} (V);   \draw[->,double, bend left,out=70,in=105] (B) to node {}
    (Y);   \draw[->,double] (B) to node {} (M);   \draw[->] (SM) to
    node {} (M);    \draw[->] (SG) to node {} (Z);   \draw[->, bend
    right] (SAB) to node {} (A);   \draw[->, bend left] (SAB) to node
    {} (B);
  \end{tikzpicture}
  \caption{\small Influence diagram for Example \ref{ex:Example4}.}
  \label{Vansteelandt}
\end{figure}

\begin{ex}
  \label{ex:Example4} Consider the class of problems described by the
  influence diagram of \figref{Vansteelandt}.  Take the variable $V$,
  which is a putative common direct influence on $G$ and $Y$, but not
  a direct influence on $M$, to be unobserved.  Let all other
  variables in the diagram be observed, with $A, B$ and $Y$ binary.
   \newpar Consider the question whether the total effects of $A$ and
  $B$ interact mechanistically. This question can be addressed on the
  basis of the data in this example. To see this, consider that, with
  the choice $W \equiv \emptyset$, the causal conditions for the
  identifiability of $A*B$, are satisfied. This is because with that
  choice \condref{caus2} becomes vacuous, \condref{caus3} does not
  apply since $A$ and $B$ are binary, and \condref{caus1} in this case
  takes the form $Y \indep \sigma_{AB} \mid (A,B)$, which follows from
  the graph in \figref{Vansteelandt}. We conclude that, whenever the
  remaining (structural and uniform positivity) conditions for
  identifiability hold, the presence of superadditivity in the sense
  of
  \begin{equation*}
    R_{11}-R_{01}-R_{10}+R_{00}>0,
  \end{equation*}
  \noindent or of excess risk if appropriate, implies $A*B$ in this
  example.  \newpar Now consider the question whether the effects of
  $A$ and $B$ interact mechanistically in a specific stratum $M=m,
  G=g$ of the population.  This question can be addressed on the basis
  of the data in this example. To see this, consider that, with the
  choices $W \equiv C \equiv (M,G)$ and $U \equiv \emptyset$, the
  causal conditions for the identifiability of $A*B\, [M,G]$, are
  satisfied. This is because, with those choices, \condref{caus2}
  becomes vacuous and \condref{caus3} can be dropped on the grounds
  that $A$ and $B$ are binary, while \condref{caus1} takes the form $Y
  \indep \sigma_{AB} \mid (M,G,A,B)$, which holds in the graph.  We
  conclude that, whenever the remaining (structural and uniform
  positivity) conditions for identifiability hold, the presence of
  superadditivity in the sense of
  \begin{equation}
    R_{11mg}-R_{01mg}-R_{10mg}+R_{00mg}>0,
    \label{example4}
  \end{equation}
  \noindent or of excess risk if appropriate, implies $A*B\, [M,G]$ in
  this example.

  \newpar Next consider the question whether the direct effects of $A$
  and $B$ on $Y$, unmediated by $M$, interact mechanistically in a
  stratum $G=g$ of the population.  This question can be addressed on
  the basis of the data in this example.  This is because
  \condref{struct7} for the equivalence of  the interaction $A*B\,
  [M,G]$ (which we have proved testable in this example) and $A*B
  \,\mid\, M [G]$  takes the form $Y \indep \sigma_{M} \mid
  (A,B,M,G,\sigma_{AB})$, which follows from the graph in
  \figref{Vansteelandt}. We conclude that a test of the
  superadditivity condition \eqref{example4} will  test the hypothesis
  that the direct effects of $A$ and $B$ on $Y$, unmediated by $M$,
  interact mechanistically in a stratum $G=g$ of the population. The
   result of the test will, in general, depend on the  chosen values
  for $m$ and $g$.   
\end{ex}

  \begin{ex}
    \label{ex:Example5} Consider the class of problems described by
    the influence diagram of \figref{fig3}, and take all the variables
    in this diagram to be observed. In this example, the $\sigma_{AB}
    \rightarrow G$ arrow indicates that the  probability distribution
    of $G$ may depend on whether the values of $A$ and $B$ are
    generated observationally or interventionally.

    \newpar The question whether the total effects of $A$ and $B$
    interact mechanistically, in the sense of $A*B$, {cannot} be
    addressed on the basis of the data in this example, the reason
    being that \condref{caus1} for this interaction to be identifiable
    takes the form $Y \indep \sigma_{AB} \mid (A,B)$, which does not
    hold in this case.  The culprit here is the $\sigma_{AB}
    \rightarrow G$ arrow.

    \newpar But consider the question whether the effects of $A$ and
    $B$ interact in a specific stratum $G=g$ of the population.  This
    question can be addressed on the basis of the data in this
    example. To see this, consider that, with the choices $W \equiv G
    \equiv C$ and $U \equiv \emptyset$ the causal conditions for the
    identifiability of $A*B\, [G]$, are satisfied. This is because,
    with those choices, \condref{caus2} becomes vacuous and the
    remaining causal conditions take the forms $Y \indep \sigma_{AB}
    \mid (G,A,B)$ and $A \indep B \mid (G, \sigma_{AB})$, both of
    which hold in the graph.  We conclude that, whenever the remaining
    (structural and uniform positivity) conditions for identifiability
    hold, the presence of superadditivity in the sense of
    \begin{equation}
      R_{11g}-R_{01g}-R_{10g}+R_{00g}>0,
      \label{example5}
    \end{equation}
    \noindent or of excess risk if appropriate, implies $A*B\, [G]$ in
    this example.

    \newpar Next consider the question whether the direct effects of
    $A$ and $B$ on $Y$, unmediated by $G$, interact mechanistically.
    This question can be addressed on the basis of the data in this
    example.  This is because \condref{struct7} for the equivalence of
     the interaction $A*B\, [G]$ (which we have proved testable in
    this example) and $A*B \mid G$  takes the form $Y \indep
    \sigma_{G} \mid (A,B,G,\sigma_{AB})$, which follows from the graph
    in \figref{Vansteelandt}.  We conclude that a test of the
    superadditivity condition \eqref{example5} will test the
    hypothesis that the direct effects of $A$ and $B$ on $Y$,
    unmediated by $G$, interact mechanistically.
  \end{ex}

  \todo[inline]{We may wish to combine the diagrams of the last two
    figures into a single figure.}

\begin{figure}[bt]
  \centering
  \begin{tikzpicture}[align=center,node distance=2cm]
      \node (Y) [circular drop shadow,decorate, fill=blue!20,circle]
      {$Y$};   \node (Z) [circular drop shadow,below left
    of=Y,decorate,   fill=blue!20,circle] {$G$};   \node (M) [circular
    drop shadow,below right of=Z,   decorate, fill=blue!20,circle]
    {$M$};   \node (A)[circular drop shadow,below right of=M,decorate,
      fill=blue!20,circle] {$A$};   \node (B) [circular drop
    shadow,below left of=Z,   decorate, fill=blue!20,circle] {$B$};
      \node (SM) [left of=A] {$\sigma_{M}$};   \node (SG) [left of=Z,
    node distance=1.5cm] {$\sigma_{G}$};   \node (SAB) [left of=SM]
    {$\sigma_{AB}$};
      \draw[->,double, bend right] (M) to node {} (Y);
      \draw[->,double](Z) to node {} (Y);
      \draw[->,double] (Z) to node {} (M);   \draw[->,double] (A) to
    node {} (M);   \draw[->,double, bend right,   in=-140] (A) to node
    {} (Y);   \draw[->,double, bend left, out=90,   in=90] (B) to node
    {} (Y);   \draw[->,double] (B) to node {} (M); \draw[->,double]
    (B) to node {} (Z);   \draw[->] (SM) to node {} (M);  \draw[->]
    (SG) to node {} (Z);   \draw[->, bend right] (SAB) to node {} (A);
       \draw[->, bend right] (SAB) to node {} (Z);   \draw[->, bend
    left] (SAB) to node {} (B);
  \end{tikzpicture}
  \caption{\small Causal diagram for Example \ref{ex:Example5}}
  \label{fig3}
\end{figure}
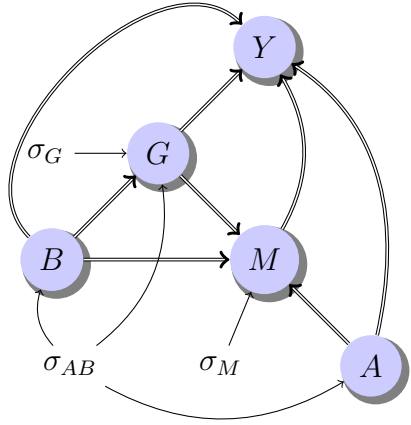

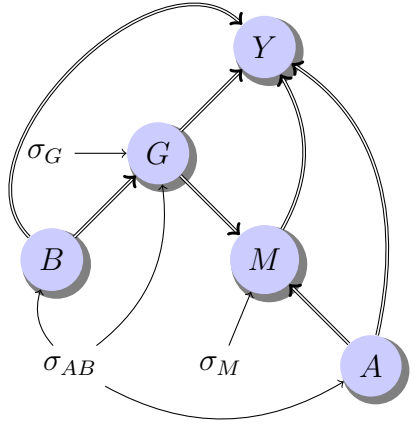
\begin{figure}[bt]
  \centering
  \begin{tikzpicture}[align=center,node distance=2cm]
      \node (Y) [circular drop shadow,decorate, fill=blue!20,circle]
      {$Y$};   \node (Z) [circular drop shadow,below left
    of=Y,decorate,   fill=blue!20,circle] {$G$};   \node (M) [circular
    drop shadow,below right of=Z,   decorate, fill=blue!20,circle]
    {$M$};   \node (A)[circular drop shadow,below right of=M,decorate,
      fill=blue!20,circle] {$A$};   \node (B) [circular drop
    shadow,below left of=Z,   decorate, fill=blue!20,circle] {$B$};
      \node (SM) [left of=A] {$\sigma_{M}$};   \node (SG) [left of=Z,
    node distance=1.5cm] {$\sigma_{G}$};   \node (SAB) [left of=SM]
    {$\sigma_{AB}$};
      \draw[->,double, bend right] (M) to node {} (Y);
      \draw[->,double](Z) to node {} (Y);
      \draw[->,double] (Z) to node {} (M);   \draw[->,double] (A) to
    node {} (M);   \draw[->,double, bend right,   in=-140] (A) to node
    {} (Y);   \draw[->,double, bend left, out=90,   in=90] (B) to node
    {} (Y);
    \draw[->,double] (B) to node {} (Z);   \draw[->] (SM) to node {}
    (M);  \draw[->] (SG) to node {} (Z);   \draw[->, bend right] (SAB)
    to node {} (A);    \draw[->, bend right] (SAB) to node {} (Z);
      \draw[->, bend left] (SAB) to node {} (B);
  \end{tikzpicture}
  \caption{\small Causal diagram for Example \ref{ex:Example6}}
  \label{fig4}
\end{figure}

\begin{ex}
  \label{ex:Example6} Finally consider the class of problems described
  by the influence diagram of \figref{fig4}, and take all the
  variables in this diagram to be observed.

  \newpar Consider the question whether the effects of $A$ and $B$
  interact in a specific stratum $M=m,G=g$ of the population.  This
  question can be addressed on the basis of the data in this
  example. To see this, consider that, with the choices $W \equiv C
  \equiv (M,G)$ and $U \equiv \emptyset$,  the causal conditions for
  the identifiability of $A*B [M,G]$, are satisfied. This is because,
  with those choices, \condref{caus2} becomes vacuous and the
  remaining causal conditions take the forms $Y \indep \sigma_{AB}
  \mid (A,B,M,G)$ and $A \indep B \mid (M,G, \sigma_{AB})$, both of
  which hold in the graph of \figref{fig4}.  We conclude that,
  whenever the remaining (structural and uniform positivity)
  conditions for identifiability hold, the presence of superadditivity
  in the sense of
  \begin{equation*}
    R_{11mg}-R_{01mg}-R_{10mg}+R_{00mg}>0,
  \end{equation*}
  \noindent or of excess risk if appropriate, implies $A*B [M,G]$ in
  this example.

  \newpar Next consider the question whether the direct effects of $A$
  and $B$ on $Y$, unmediated by $M$ or $G$, interact mechanistically.
  This question can be addressed on the basis of the data in this
  example.  This is because \condref{struct7} for the equivalence of
   the interaction $A*B [M,G]$ (which we have proved testable in this
  example) and $A*B \mid (M,G)$  takes the form $Y \indep
  (\sigma_G,\sigma_M) \, \mid \, (A,B,G,M,\sigma_{AB})$,  which
  follows from the graph in \figref{fig4}.
\end{ex}

\section{Causality and agency}
\label{Philosophical considerations}

Some of the above examples raise some issues of the interpretation of
``causality'' in our approach. According to our description so far,
that concept has been closely tied to the possibility of making
external interventions to set values for the ``causal variables'' $A$
and $B$. This conception is in line with philosophical ``agency''
theories of causality, \cite{price:bjps91,hausman:book,woodward:book},
which regard causes as handles for manipulating effects. However, such
an anthropocentric manipulationist view is unnecessarily restrictive,
and can hamper application of causal inference to numerous scientific
disciplines that demand a more general notion of cause, not tied
simply to what human agents can do.

\newpar In the first two examples above, while the variable $A =$
``stated ethnicity'' was manipulable (and was manipulated), we can not
reasonably regard the variable $B=$ ``age of participant'' as
manipulable. We might however conceive of being able to observe an
individual at various points of her life, and be interested in the way
in which her age then might make a difference to her psychological
response to certain media framing techniques. Psychologists have
knowledge, theories and hypotheses about the role of age in the
response process. They can, for example, make informed guesses about -
and explain on the basis of psychological theories - the different
outcome we might have observed had the individual been younger or
older than he is ({\em e.g.\/}, ``young people tend to react with less
anxiety''). Specific psychological mechanisms and reactions are
associated with young age. We should not give up looking into them
simply because the age variable falls outside the standard
manipulability theory of causation.

\newpar As another example, in epidemiology it is often appropriate to
consider, as a cause of a disease, a variable such as genotype, whose
manipulation by human beings is not practically possible; and
application of mechanistic interaction tests to investigations of
epistasis or pharmacogenomics will require a broader conception of
``intervention'' than the agency approach typically supplies.  Recent
discussions of the topic \citep{woodward2013} have loosened the strict
confines of the manipulationist theory, regarding as an
``intervention'' any appropriate (in a sense that has to be made
clear) exogenous causal process, without any necessary connection with
human action.


\section{Related work}
\label{Related work}

A recent paper by \cite{VanderWeeleStochastic} (hereafter VR) tackles
mechanistic interaction via stochastic (rather than deterministic)
potential outcomes (POs).  In the standard PO formulation, the value
that $Y$ would take in individual $\omega$ in response to an
intervention that sets $(A,B)$ to values $(a,b)$ is regarded as a
potential outcome, $Y_{ab}(\omega)$.  Potential outcomes are fixed for
each particular individual even before the treatment is applied, and
unaffected by the particular regime in which the values of $A$ and $B$
are set.  VR relax this by allowing each individual $\omega$ to be
characterized by a {\em stochastic\/} potential outcome,
$Y_{ab}(\omega)$, that varies in the individual according to a
Bernoulli distribution with the expected value fixed by the
intervention and by random circumstances, although these latter are
assumed {not} to be affected by the treatment.  Because of the latter
constraint, it is not clear whether VR's approach, as currently
formulated, copes with situations where a stochastic mediator of the
effect of $(A,B)$ on $Y$ introduces intervention-dependent random`.
variation\footnote{We also note that in Rubin's standard PO
  formulation there is a value of the response for each individual and
  possible intervention, and such value is constant across all
  possible regimes, in the sense that it is not affected by the way
  the values of $A$ and $B$ are generated.  In VR's approach, the
  response has its expected value fixed by the particular individual,
  set of random circumstances and intervention.  But conditional on
  this expected value, is the realized value of the response assumed
  to vary across regimes?  In other words, is the observationally
  detected response identical to what I would have observed had I
  fixed the same treatment by intervention? And, if the answer to the
  above question is negative, are the regime-specific versions of the
  response assumed independent? We feel that the question matters to
  the very purpose of carrying inferences from the observational to
  other regimes.  These considerations are related to certain
  ambiguities of counterfactual-based formulations of causality
  \citep{dawidjasa}}.

\newpar \cite{Ramsahai2013} gives a fully probabilistic account of
mechanistic interaction, which boasts aspects of greater generality
relative to ours, including freedom from monotonicity assumptions
about the effects of $A$ and $B$. There are also aspects of lesser
generality: no attempt is made in Ramsahai' paper to examine the
implications of the presence of continuous causal factors. It is
therefore appropriate to proceed by comparing Ramsahai's method and
ours in the special case where $A$ and $B$ are binary variables, with
$(a,a',b,b') \equiv (1,0,1,0)$. In this special case, our condition
for presence of mechanistic interaction, as expressed by
(\ref{synergism 5}), specializes to \begin{equation}
  \label{synergism Ram}
  \frac{\pi_w(1,1)}{\pi_{w}(0,1)} <
  \frac{\pi_w(1,0)}{\pi_{w}(0,0)}.
\end{equation}
\noindent As seen in \S \ref{mechint}, under Condition~\ref{struct2},
Equation~(\ref{synergism Ram}) implies strictly positive monotonicity
of the effects of $A$ and $B$ upon $Y$, as expressed by the
inequalities
\begin{eqnarray}
  \label{synergism 1}
  \pi_w(1,1) &< \pi_w(1,0),\\
  \label{synergism 2}
  \pi_w(1,1) &< \pi_w(0,1).
\end{eqnarray}
\noindent In our approach, these inequalities are consistent with, but
not sufficient for, the presence of mechanistic interaction.  In fact,
consistently with our concluding remarks of \S~\ref{secpoint},
inequalities \eqref{synergism 1}-- \eqref{synergism 2} do not imply
(\ref{synergism Ram}).  By contrast, in Ramsahai's approach, those
inequalities are taken to define mechanistic interaction for binary
variables.  Hence Ramsahai's definition of mechanistic interaction is
weaker than ours.  The more exacting nature of our definition of
mechanistic interaction, combined with allowance for continuity,
explains the stronger assumptions required in our approach compared to
those of Ramsahai.

\newpar To elucidate the differences between the approaches, suppose
that, in the bowling example, $A$ (the first player's ball size) takes
value 0 (the player has no ball to throw) or 1 (the player throws a
ball).  Interpret $B$ analogously.  It then seems reasonable, on the
basis of physics and common sense, to assume that inequalities
\eqref{synergism 1}--\eqref{synergism 2} hold in this example.  In
Ramsahai's formulation, this is sufficient to conclude in favour of
mechanistic interaction between the effects of the throws of the two
players, even before looking into the data, and even if the two
players act independently.  This appears to clash with our
psychological notion of synergism.  By contrast, in our formulation,
\rosso in order to conclude in favour of mechanistic interaction,
conditions \eqref{synergism 1}--\eqref{synergism 2} are not sufficient
because they do not contradict the intuitive idea of independent
throws expressed by \eqref{pointnull}.  \nero

\newpar To conclude, we note that our approach uses statistics (excess
risk and superadditivity) which are often testable at negligible
computational cost in prospective studies, and (approximately and
under assumptions) also in retrospective studies.  By contrast,
attention needs to be paid to the computational feasibility of
Ramsahai's approach.

\section{Discussion}
\label{Discussion}

Mechanistic interaction has often been tackled within a potential
outcome framework \citep{Rubin1974} or within an equivalent
formulation of causality.  We have discussed possible limitations of
this approach.  We have also discussed limitations of current
approaches to mechanistic interaction which reject the potential
outcome formulation in favour of the standard probability formalism.
Motivated by the limitations of the previous approaches, we have
proposed a novel definition of the causal notion of mechanistic
interaction, and presented sufficient conditions for its
identification from observational data.  Because these conditions are
expressed in terms of conditional independence, they hold irrespective
of particular parametric or distributional assumptions about the
problem variables.  A further advantage of our conditional
independence formulation of the identifiability conditions is that
these can be straightforwardly checked on a causal diagram of the
problem, when this is available.  The use of causal diagrams for the
mentioned purposes has been extensively illustrated.

\newpar Our theory provides conditions for testing for mechanistic
interaction in (real or hypothetical) situations in which an
intervention is exerted on variables (even post-treatment ones)
different from the main factors $A$ and $B$ of interest.  We have
discussed the connection between this and the idea of mechanistic
interaction between effects that flow along specific paths in a causal
diagram representation of the problem.

\newpar Importantly, our method does not require the assumption that
$Y$ depends on its causal influences in a functional way.  By relaxing
such an assumption, our method gains applicability in a much wider
range of situations, and confers more leeway on the researcher in the
choice of the conditioning variables in the test.

\newpar Once the conditions for a test of the mechanistic interaction
of interest have been found valid, the actual test involves simple
(and well known) excess risk or superadditivity statistics.  These
tests are valid under prospective sampling and (under assumptions)
retrospective sampling.  In the latter case, a key assumption is that
the response event of interest is rare under any possible
configuration of the causal factors.  In the context of retrospective
case-control studies in epidemiology, this is the well-known rare
disease assumption that typically motivates this kind of studies.

\newpar Finally, our approach embraces the very large class of
applications where the main causal factors, $A$ and $B$, are only
available as a discretized version of the fundamental variables, no
longer available in their original continuous form.

\newpar Various possible enhancements of the method are envisaged, one
of these being the extension of the theory to embrace higher-order
mechanistic interactions.  Equally important will be the application
of the method in a variety of situations and disciplines, from genetic
epidemiology ({\em e.g.\/}, in the identification of gene-environment
interactions) to experimental psychology.  We hope that the proposed
method will help researchers better to identify from data analysis
small sets of interactions underlying mechanisms of scientific
interest.

\section{Acknowledgments}
\label{Acknowledgments}

Carlo Berzuini was partially supported by the FP7-305280 MIMOmics
European Collaborative Project, as part of the HEALTH-2012- INNOVATION
scheme.

\bibliographystyle{plainnat} \bibliography{bibliografia}
\end{document}